\documentclass[11pt]{article}

\usepackage[margin=1in]{geometry}
\usepackage{amsmath, amssymb, amsthm}
\usepackage{mathtools}
\usepackage{etoolbox}
\usepackage{hyperref}
\usepackage{parskip}
\usepackage[T1]{fontenc}

\theoremstyle{plain}
\newtheorem{theorem}{Theorem}[section]

\newtheorem{corollary}[theorem]{Corollary}
\theoremstyle{definition}
\newtheorem{definition}{Definition}[section]
\newtheorem{remark}{Remark}[section]
\AtEndEnvironment{remark}{\hfill$\triangle$}

\BeforeBeginEnvironment{theorem}{\vspace{\parskip}}
\BeforeBeginEnvironment{lemma}{\vspace{\parskip}}
\BeforeBeginEnvironment{corollary}{\vspace{\parskip}}
\BeforeBeginEnvironment{definition}{\vspace{\parskip}}
\BeforeBeginEnvironment{remark}{\vspace{\parskip}}

\newcommand{\SOTM}{\mathcal{M}}
\newcommand{\oracle}{\mathcal{O}}
\newcommand{\alphabet}{\Sigma}
\newcommand{\vocab}{\mathcal{T}}
\newcommand{\tokenizer}{\tau}
\newcommand{\tok}{\mathrm{tok}}
\newcommand{\TOK}{\mathrm{TOK}}
\newcommand{\Ber}{\mathrm{Ber}}
\newcommand{\KL}[2]{D\!\left(#1 \,\middle\|\, #2\right)}
\newcommand{\E}{\mathbb{E}}

\newcommand{\Rel}{\operatorname{Rel}}

\title{\textbf{Token Complexity of Certifying Stochastic-Oracle Reliability}}
\author{Jie Wang\footnote{Richard Miner School of Computing and Information Sciences, University of Massachusetts, Lowell, MA 01854, USA.} \\
jie\rule{0.4em}{0.4pt}wang@uml.edu}
\date{}

\begin{document}

\maketitle

\begin{abstract}
Wang~\cite{Wang2026} introduced the Stochastic-Oracle Turing Machine
(SOTM) framework and defined token complexity as the minimum expected cost
of interacting with a stochastic oracle needed to attain a specified
solution quality for a task. This paper develops an analogous notion for
certifying the reliability of a stochastic oracle on a given domain.
Certification token complexity is the minimum expected token cost required,
with controlled error probability, to distinguish oracles that meet a
target reliability level from those that fall below a lower reliability
threshold.

We construct an SPRT-based certification SOTM that queries the oracle,
computes binary correctness scores, and stops when the accumulated
log-likelihood evidence crosses a decision threshold. The SOTM halts almost
surely, satisfies the desired two-sided error guarantee over the
reliability regions to be certified, and yields an explicit upper bound on
certification token complexity in terms of the reliability thresholds, the
error bound, and the expected per-turn token cost. We then establish a
matching information-theoretic lower bound: even with adaptive queries,
every error-bounded certification SOTM must incur the same leading-order
expected token cost as the SPRT-based construction as the prescribed error
bound tends to zero. Together, these bounds characterize the leading-order
certification token complexity in the small-error regime.
\end{abstract}

\section{Introduction}\label{sec:intro}

The emergence of AI-augmented computing has created a new computational
paradigm in which a program delegates knowledge-intensive or
skill-intensive subtasks to an AI system. In such settings, the program
relies on knowledge and problem-solving capabilities supplied by the AI
system and not reproducible by the program alone. The AI system may be a
single large language model, an in-house trained model, or a collection
of stochastic systems together with a control mechanism that directs
queries. Such an AI system can naturally be modeled as a stochastic
oracle whose responses are drawn from a fixed distribution for each
query~\cite{Wang2026}.

Interactions with stochastic oracles differ from calls to classical
oracles in three fundamental ways. First, for each query, a stochastic
oracle's response is drawn from a fixed, query-dependent distribution,
whereas a classical oracle always returns the same response to the same
query. Second, the correctness of a stochastic oracle's response is
probabilistic, whereas a classical oracle's response has deterministic
correctness. Third, each interaction carries an intrinsic token cost: the
tokens needed to pose a query and those needed to read the response. In
contrast, calls to a classical oracle are assumed not to incur such costs.

These costs are not merely implementation details. Query tokens encode
requests formulated by the machine, whereas response tokens encode
oracle-generated information. 
A theory of computation with stochastic
oracles must therefore account not only for whether a task can be solved,
but also for the token cost of obtaining enough useful information from the oracle.
Wang~\cite{Wang2026} formalized this perspective in the
Stochastic-Oracle Turing Machine (SOTM) framework, defining token
complexity as the minimum interaction cost required to attain a specified
solution quality for a task.\footnote{The SOTM was called the
AI-Oracle Turing Machine (AOTM) in that work; we rename it here since
AI is a historically shifting term, while stochastic precisely names the
operative mathematical property of the oracle.} 

Throughout the paper, we use \emph{oracle} as shorthand for a stochastic
oracle in the SOTM sense, unless classical oracles are explicitly being
discussed.

Wang's work treats the oracle as a fixed resource available to the SOTM
designer and studies structural relationships between cost and solution
quality~\cite{Wang2026}. This paper addresses the complementary question
of how many tokens are needed to certify, with controlled error
probability, the reliability of an existing stochastic oracle on a
specified application domain.

This question is especially important for the domain-specific deployment
of stochastic oracles. General-purpose benchmarks may not adequately
measure an oracle's performance in specialized applications, such as legal
analysis, business analysis, medical triage, scientific reasoning,
customer support, and internal enterprise workflows. Before deploying an
oracle in such a domain, one may wish to certify that its reliability meets
a target level while controlling the probability of an incorrect
certification decision.

This paper formulates this task as \emph{oracle reliability certification}.
A certification SOTM interacts with the oracle to be certified, evaluates
its responses using a fixed binary scoring rule, and determines whether the
oracle lies in the certifiably reliable or certifiably unreliable region.
No definitive decision is required in the ambiguity region between the two
thresholds.

We define certification token complexity as the minimum expected token
cost among all certification SOTMs that achieve the required
error-bounded distinction between these two certifiable regions. This
definition parallels Wang's task-level token complexity, but the goal is
different. Rather than asking how many tokens are required to solve a
task to a desired quality level, we ask how many tokens are required to
evaluate an existing oracle with controlled statistical error. In this
sense, oracle reliability certification is a measurement problem: the
oracle is treated as a black-box stochastic system, and the certification
SOTM must spend tokens to obtain enough evidence about its reliability.

To upper-bound this certification token complexity, we construct an
explicit certification SOTM based on the Sequential Probability Ratio Test (SPRT)~\cite{Wald1947,WaldWolfowitz1948,LehmannRomano2005}. 
The certification SOTM repeatedly samples evaluation queries, obtains oracle
responses, computes binary correctness scores, and updates a
log-likelihood statistic. The SOTM stops when the accumulated
evidence crosses one of two decision thresholds. We prove that this SPRT-based certification SOTM halts almost surely,
satisfies the desired two-sided error guarantee whenever the oracle's
reliability lies above the target threshold or below the lower threshold,
and gives an explicit upper bound on certification token complexity. This bound can be
estimated in advance from the reliability thresholds, the allowable error
probability, and the expected per-turn token cost.

We then establish a matching information-theoretic lower bound: even with
adaptive queries, every certification SOTM with controlled error must incur
the same leading-order expected token cost as the SPRT-based certification SOTM as
the prescribed error bound tends to zero. Together, the upper and lower
bounds characterize the leading-order certification token complexity in the
small-error regime.

The resulting characterization shows that certification becomes only
logarithmically more costly as the prescribed error bound decreases, but
quadratically more costly as the separation between the lower and target
reliability thresholds narrows. Thus, demanding greater confidence is
relatively cheap compared with demanding a sharper distinction between
nearby reliability levels. This provides a principled way to estimate the
token complexity of certifying an LLM or other AI systems before
deployment in a specialized application domain.

The remainder of the paper is organized as follows.
Section~\ref{sec:preliminaries} recalls the necessary definitions from
the SOTM framework, including stochastic oracles, interaction cost, and
task-level token complexity. Section~\ref{sec:certification setup} introduces oracle
reliability certification and defines certification token complexity.
Section~\ref{sec:upper}
constructs the SPRT-based certification SOTM and proves an upper bound on
certification token complexity. Section~\ref{sec:lower} proves the
matching information-theoretic lower bound and explains the resulting
token-complexity characterization. Section~\ref{sec:conclusion}
concludes the paper.

\section{Preliminaries}\label{sec:preliminaries}

We use the SOTM model and token-complexity framework of
Wang~\cite{Wang2026}, recalling only the definitions needed here. Fix an
encoding over a finite alphabet~$\alphabet$, and regard all objects as
strings in~$\alphabet^*$.

A \emph{Stochastic-Oracle Turing Machine (SOTM)} is a pair
$\SOTM = (M, \oracle)$, where~$M$ is a probabilistic Turing machine
and~$\oracle$ is a stochastic oracle, a family
$\{\mathcal{D}_q : q \in \alphabet^*\}$ of response distributions,
one for each query string. The machine~$M$ has an input tape, an
output tape, one or more work tapes, a read-only random tape, a query tape, and a
response tape. On each oracle call,~$M$ writes a query string
$q \in \alphabet^*$ to the query tape; $\oracle$ draws a response
$r \in \alphabet^*$ from~$\mathcal{D}_q$ and writes it to the
response tape.

A \emph{turn}\footnote{We use the term \emph{turn}, which is standard in
LLM interactions, for what is often called a \emph{round} in sequential
analysis.} is one complete interaction cycle between $M$ and $\oracle$:
the local computation by which $M$ produces a query, the oracle call that
returns a sampled response, and the subsequent local computation by which
$M$ reads and processes that response.

The oracle provides response distributions that are external to the
machine~$M$. Thus, although $M$ can obtain a response
$r\sim\mathcal{D}_q$ by querying $\oracle$ on $q$, the model does not
assume that $M$ can reproduce these distributions internally. This
asymmetry makes oracle access nontrivial and gives rise to the reliability
questions studied below.

A \emph{task} is a tuple $T = (X, Y, S, \mathcal{D}_X)$
with input space $X \subseteq \alphabet^*$, output space
$Y \subseteq \alphabet^*$, score function
$S\colon X \times Y \to [0,1]$ (polynomial-time
computable in $|x|+|y|$), and input distribution~$\mathcal{D}_X$.

A tokenizer $\tokenizer\colon \alphabet^* \to \vocab^*$~\cite{Sennrich2016,Schuster2012,Kudo2018}
maps queries and responses to token strings. The token count at
turn~$i$, with query~$q_i$ and response~$r_i$, is
$|\tau(q_i)| + |\tau(r_i)|$. The total token count of~$\SOTM(x)$
over~$n$ turns is
\[
  \tok_\SOTM(x)
  \;=\; \sum_{i=1}^n |\tau(q_i)| + \sum_{i=1}^n |\tau(r_i)|,
\]
where we write $\tok_{\SOTM,Q}(x) = \sum_{i=1}^n |\tau(q_i)|$ and
$\tok_{\SOTM,R}(x) = \sum_{i=1}^n |\tau(r_i)|$ for the total query
and response token counts. 

We allow query and response tokens to have different costs, writing
$\alpha$ for the cost of a query token and~$\beta$ for the cost of a
response token. This distinction reflects the practical asymmetry that
processing a query token and generating a response token can incur
different intrinsic costs; in typical LLM deployments, response tokens
are more expensive, so one often has~$\beta \ge \alpha$.

Fix token cost parameters $\alpha,\beta>0$. For an SOTM~$\SOTM$, let
$\tok_{\SOTM,Q}$ denote the total number of query tokens issued by the
machine, and let $\tok_{\SOTM,R}$ denote the total number of response
tokens returned by the oracle. The \emph{total token cost} of~$\SOTM$ is
\[
  \TOK_{\SOTM}^{\alpha,\beta}
  =
  \alpha\,\tok_{\SOTM,Q}
  +
  \beta\,\tok_{\SOTM,R}.
\]
When the cost parameters are clear from context, we write
$\TOK_{\SOTM}$ for $\TOK_{\SOTM}^{\alpha,\beta}$.

For an input-dependent task, we write
$\TOK_{\SOTM}^{\alpha,\beta}(x)$ for the total token cost incurred by
$\SOTM$ on input~$x$. Its expected token cost on a task distribution
$\mathcal{D}_X$ is
\[
  \E_{x\sim\mathcal{D}_X}
  \bigl[\TOK_{\SOTM}^{\alpha,\beta}(x)\bigr].
\]
For a certification SOTM, which takes no external input, the expected
token cost is
\[
  \E\bigl[\TOK_{\SOTM}^{\alpha,\beta}\bigr],
\]
where the expectation is over the query choices, oracle responses,
computed scores, stopping time, final verdict, and any internal
randomness of the machine.

The \emph{task-level token complexity} of~$T$ at quality threshold
$\theta\in[0,1]$ and token cost parameters $\alpha,\beta>0$ is
\[
  \kappa_T(\theta;\alpha,\beta)
  =
  \inf_{\SOTM}
  \E_{x\sim\mathcal{D}_X}
  \bigl[\TOK_\SOTM^{\alpha,\beta}(x)\bigr],
\qquad \text{subject to:}
\qquad
  \E_{x\sim\mathcal{D}_X}
  \bigl[S(x,\SOTM(x))\bigr]
  \ge
  \theta.
\]

\paragraph{Convention.}
Throughout the paper, whenever they are used, $\alpha$ and $\beta$ denote
the costs per query token and per response token, respectively.


\section{Reliability Certification Setup}\label{sec:certification setup}

To evaluate an oracle $\oracle = \{\mathcal D_q \mid q \in \Sigma^*\}$ on a task, the explicit output space~$Y$ of the task is not
needed, as that output space belongs to the full task-solving problem, in
which an SOTM produces a final answer using its internal randomness and
probabilistic algorithms. In the certification setting, by contrast, we
evaluate the oracle's response behavior directly by applying a scoring
function to each query--response pair.

We therefore define an \emph{evaluation domain}
$\boldsymbol{d} = (X, S, \mathcal{D}_X)$, consisting of a query
domain~$X$, a distribution~$\mathcal{D}_X$ over queries, and a binary
scoring function
\[
S : X \times \alphabet^* \to \{0,1\}.
\]
The scoring function may be hand-specified or implemented by a fixed
evaluator model; in either case, it is treated as fixed throughout
certification.

We assume that $X$ is finite and serves as the sample space, and that
oracle-response lengths are controlled relative to the lengths of their
corresponding queries; namely,
there are constants $\lambda,c<\infty$ such that every response~$r$ to a
query $x\in X$ satisfies
\[
  |\tau(r)| \le \lambda\,|\tau(x)|+c.
\]
This captures the practical setting of certification on a fixed
evaluation set with a bounded response protocol, such as an LLM
evaluation with a prescribed maximum output length. Since $X$ is finite, the query length $|\tau(x)|$ is uniformly bounded
over $x\in X$. The length-control assumption then uniformly bounds
$|\tau(r)|$ as well. Consequently, the per-turn token cost is uniformly
bounded.

The underlying implementation of~$\oracle$, whether a single model, an
ensemble, or another mechanism, is abstracted away. The certification
SOTM has access only to the oracle's query--response behavior, represented
by the response distributions associated with queries $x\in X$.

The \emph{reliability} of an oracle~$\oracle$ on an evaluation domain
$\boldsymbol{d}$ is
\[
  \Rel_{\boldsymbol{d}}(\oracle)
  =
  \E_{x\sim\mathcal{D}_X}\!\left[
    \E_{r\sim\mathcal{D}_x}\!\bigl[S(x,r)\bigr]
  \right].
\]
When the oracle and evaluation domain are clear from context, we write
\[
  p=\Rel_{\boldsymbol{d}}(\oracle).
\]
We sometimes abbreviate
$
\E_{x\sim\mathcal{D}_X}\!\left[
  \E_{r\sim\mathcal{D}_x}\!\bigl[S(x,r)\bigr]
\right]
$
as
$
\E_{x\sim\mathcal{D}_X,\, r\sim\mathcal{D}_x}\!\bigl[S(x,r)\bigr]
$.

\begin{definition}[Certification Thresholds]
Fix $0 < p_0 < p_1 < 1$. On a given evaluation domain, an
oracle~$\oracle$ is said to \emph{meet the target level} if its
reliability satisfies $p \ge p_1$, and to \emph{fall below the baseline
level} if $p \le p_0$. The interval $(p_0,p_1)$ is the
\emph{ambiguity gap}: when $p \in (p_0,p_1)$, the oracle is not required
to be certified in either direction. The thresholds $p_0$ and $p_1$
therefore partition the certifiable cases, while the gap is an
indifference region in which no finite-sample rule is required to
decide.
\end{definition}

\paragraph{Convention.}
Unless otherwise stated, all definitions and results throughout the paper
are formulated for a fixed oracle $\oracle$, a fixed evaluation domain
$\boldsymbol{d}=(X,S,\mathcal{D}_X)$, and fixed reliability thresholds
$0<p_0<p_1<1$.

\begin{definition}[Certification SOTM]
A \emph{certification SOTM} is an SOTM $\SOTM=(M,\oracle)$ in which the
machine $M$ receives no external input. The scoring function $S$ and
the distribution $\mathcal{D}_X$ are known to $M$, whereas the
reliability
\[
p = \Rel_{\boldsymbol{d}}(\oracle)
\]
is unknown to $M$.

The machine proceeds sequentially. At turn $i$, it selects a query
$x_i\in X$, possibly as a function of the preceding transcript, queries
$\oracle$ with $x_i$, receives a response $r_i$, and computes the
binary score
\[
Z_i = S(x_i,r_i) \in \{0,1\}.
\]
The transcript after $i$ turns is
\[
\mathcal{T}_i = (x_1,r_1,\ldots,x_i,r_i).
\]

The machine is specified by a stopping time $N$ and a verdict function.
After observing $\mathcal{T}_i$, the machine decides whether to halt;
$N$ is the first turn at which it halts. If $N<\infty$, the verdict
function maps the stopped transcript to
\[
v(\mathcal{T}_N)\in
\{\mathsf{Reliable},\mathsf{Unreliable}\}.
\]
The verdict $\mathsf{Reliable}$ certifies that $p\ge p_1$, whereas
$\mathsf{Unreliable}$ certifies that $p\le p_0$. A run with
$N=\infty$ produces no verdict. No correctness requirement is imposed
when $p\in(p_0,p_1)$.
\end{definition}

\begin{definition}[Adaptive and non-adaptive certification SOTMs]
A certification SOTM is \emph{adaptive} if its choice of query~$x_i$ may
depend on the previous transcript~$\mathcal{T}_{i-1}$. It is
\emph{non-adaptive} if the queries are chosen independently of the
previous oracle responses and computed scores. 
\end{definition}

\begin{definition}[Error-bounded certification SOTM]
\label{def:error-bounded}
Fix an error parameter
$\varepsilon\in(0,1/2)$. A certification SOTM is
\emph{$\varepsilon$-error-bounded} if it halts almost surely and its
final verdict $v\in\{\mathsf{Reliable},\mathsf{Unreliable}\}$ satisfies
\[
  \Pr_p(v=\mathsf{Unreliable}) \le \varepsilon
  \qquad
  \text{for every } p\ge p_1,
\]
and
\[
  \Pr_p(v=\mathsf{Reliable}) \le \varepsilon
  \qquad
  \text{for every } p\le p_0.
\]
Here $\Pr_p$ denotes probability over the query choices, oracle
responses, computed scores, stopping time, final verdict, and the
machine's internal randomness when the oracle reliability on the
evaluation domain is~$p$.
\end{definition}

\paragraph{Convention.}
Throughout the paper, $\varepsilon\in(0,1/2)$ denotes the error bound of
the underlying SOTM whenever it is used.

For a certification SOTM, the \emph{sample size} is the number of oracle
queries issued before the machine halts. Equivalently, it is the number
of turns, since each turn includes one query--response interaction and
the local computation used to process it.
The
\emph{sample complexity} of the certification problem is the infimum of
$\E[N]$ over all certification SOTMs satisfying the prescribed error
guarantees.

For a certification SOTM with stopping time $N$, its expected sample size
is denoted by $\E[N]$.

\begin{definition}[Certification token complexity]
\label{def:certification-token-complexity}
The \emph{certification token complexity} is
\[
  \kappa_{\boldsymbol{d}}^{\mathrm{cert}}
  (p_0,p_1,\varepsilon;\alpha,\beta)
  =
  \inf_{\SOTM}
  \E[\TOK_\SOTM^{\alpha,\beta}],
\]
where the infimum is taken over all certification SOTMs that halt almost
surely and whose verdict $v\in\{\mathsf{Reliable},\mathsf{Unreliable}\}$
satisfies
\[
  \Pr_p(v=\mathsf{Unreliable})\le \varepsilon
  \qquad
  \text{for every } p\ge p_1,
\]
and
\[
  \Pr_p(v=\mathsf{Reliable})\le \varepsilon
  \qquad
  \text{for every } p\le p_0.
\]
\end{definition}

\section{Constructing a Certification SOTM}\label{sec:upper}

Let $p$ denote the reliability of $\oracle$ on
$\boldsymbol{d}$. 
We construct a non-adaptive certification SOTM
\[
  \SOTM_{\mathrm{SPRT}}(p_0,p_1,\varepsilon)=(M,\oracle)
\]
based on the SPRT.
The SOTM tests the hypotheses
\[
  H_0:p\le p_0
  \quad\text{versus}\quad
  H_1:p\ge p_1
\]
by updating a log-likelihood ratio after each observation and halting
once the accumulated evidence crosses a threshold.

The computation proceeds in turns. At each turn $i$, $M$
draws a query $x_i\sim\mathcal{D}_X$, independently of all previous
turns, and issues $x_i$ to $\oracle$. It receives a response
$r_i\sim\mathcal{D}_{x_i}$, independently of all previous queries and
responses, and computes
\[
  Z_i=S(x_i,r_i).
\]
Thus, $Z_1,Z_2,\ldots$ are i.i.d.\ Bernoulli random variables with
parameter $p$. That is,
\[
  Z_i\sim\Ber(p),
  \qquad
  \Pr[Z_i=1]=p,
  \qquad
  \Pr[Z_i=0]=1-p.
\]
Equivalently, for $z\in\{0,1\}$,
\[
  \Ber(p)(z)=p^z(1-p)^{1-z}.
\]

The SPRT compares the likelihoods corresponding to the boundary values
$p=p_0$ and $p=p_1$. For a realized score sequence
$z_1,\ldots,z_n$, the probabilities of observing that sequence under
these two values are, respectively,
\[
  \prod_{i=1}^n \Ber(p_0)(z_i)
  \qquad\text{and}\qquad
  \prod_{i=1}^n \Ber(p_1)(z_i).
\]
Therefore, the log-likelihood ratio in favor of $p=p_1$ over $p=p_0$
after $n$ observations is
\[
  L_n
  =
  \sum_{i=1}^{n}
  \ln\frac{\Ber(p_1)(Z_i)}{\Ber(p_0)(Z_i)}
  =
  \sum_{i=1}^{n}
  \left(
    Z_i\ln\frac{p_1}{p_0}
    +(1-Z_i)\ln\frac{1-p_1}{1-p_0}
  \right).
\]

The KL divergence between the Bernoulli distributions $\Ber(p_1)$ and
$\Ber(p_0)$ is
\[
  D\!\left(\Ber(p_1)\,\|\,\Ber(p_0)\right)
  =
  p_1\ln\frac{p_1}{p_0}
  +(1-p_1)\ln\frac{1-p_1}{1-p_0}.
\]
For brevity, write
\[
  \KL{p_1}{p_0}
  :=
  D\!\left(\Ber(p_1)\,\|\,\Ber(p_0)\right).
\]
This is the expected one-step increment of $L_n$ when $p=p_1$.
Similarly, when $p=p_0$, the expected one-step decrement is
$-\KL{p_0}{p_1}$.

This is the expected one-step increment of $L_n$ when $p=p_1$.
Similarly, when $p=p_0$, the expected one-step increment is
$-\KL{p_0}{p_1}$.

The machine computes $L_n$ recursively. Starting from $L_0=0$, it
updates, at turn $i$,
\begin{align}
  L_i &= L_{i-1}+\ell_i, \label{eq:llr-update}\\
  \ell_i
  &=
  \ln\frac{\Ber(p_1)(Z_i)}{\Ber(p_0)(Z_i)} \label{eq:ell}\\
  &=
  Z_i\ln\frac{p_1}{p_0}
  +(1-Z_i)\ln\frac{1-p_1}{1-p_0}.
  \nonumber
\end{align}
Thus, $M$ adds $\ln(p_1/p_0)>0$ after a correct response
($Z_i=1$) and adds $\ln((1-p_1)/(1-p_0))<0$ after an incorrect response
($Z_i=0$). The statistic $L_i$ accumulates evidence for whether
$p\ge p_1$ or $p\le p_0$, without requiring an explicit estimate of $p$.

Define the log-likelihood thresholds
\[
  A=\ln\frac{1-\varepsilon/2}{\varepsilon/2},
  \qquad
  B=-A,
\]
and the stopping time
\[
  N=\inf\{n\ge 1:L_n\ge A \text{ or } L_n\le B\}.
\]
Upon halting, the machine returns $\mathsf{Reliable}$ if $L_N\ge A$ and
$\mathsf{Unreliable}$ if $L_N\le B$.

\begin{theorem}[Almost-sure halting]\label{thm:sprt-halting}
The certification SOTM
$\SOTM_{\mathrm{SPRT}}(p_0,p_1,\varepsilon)$ halts almost surely under
every $p\in[0,1]$.
\end{theorem}

\begin{proof}
Since $Z_1,Z_2,\ldots$ are i.i.d.\ Bernoulli random variables with
parameter $p$, the increments $\ell_1,\ell_2,\ldots$ are also i.i.d.
For $p\in[0,1]$, define
\[
  f(p):=\E_p[\ell_1].
\]
By the definition of $\ell_1$,
\[
  f(p)
  =
  p\ln\frac{p_1}{p_0}
  +(1-p)\ln\frac{1-p_1}{1-p_0}.
\]
The function $f$ is affine and strictly increasing, since
\[
  f'(p)
  =
  \ln\frac{p_1(1-p_0)}{p_0(1-p_1)}
  >0.
\]
Moreover,
\[
  f(p_1)
  =
  \KL{p_1}{p_0}
  >0,
  \qquad
  f(p_0)
  =
  -\KL{p_0}{p_1}
  <0.
\]

Thus, $f(p)>0$ for every $p\ge p_1$ and $f(p)<0$ for every $p\le p_0$.
For $p\in(p_0,p_1)$, the mean increment may be positive, negative, or zero;
the zero case occurs only at the unique point $p^*$ defined below.

Suppose first that $f(p)\neq 0$. By the strong law of large numbers,
\[
  \frac{L_n}{n}
  =
  \frac{1}{n}\sum_{i=1}^n\ell_i
  \longrightarrow
  f(p)
  \qquad\text{almost surely}.
\]
If $f(p)>0$, then $L_n\to+\infty$ almost surely, whereas if $f(p)<0$,
then $L_n\to-\infty$ almost surely. Since the SOTM continues only while
$L_n\in(B,A)$, it follows that $N<\infty$ almost surely.

It remains to consider the case $f(p)=0$. Since $f(p_0)<0<f(p_1)$ and
$f$ is strictly increasing, there is a unique $p^*\in(p_0,p_1)$ such that
$f(p^*)=0$. Thus, under $p=p^*$, the increments of $L_n$ are bounded,
nonconstant, i.i.d.\ random variables with mean zero. By the oscillation
property of one-dimensional zero-mean random walks,
\[
  \limsup_{n\to\infty}L_n=+\infty,
  \qquad
  \liminf_{n\to\infty}L_n=-\infty
  \qquad\text{almost surely};
\]
see, for example, \cite{Feller1971}. Consequently, $L_n$ cannot remain
in the finite interval $(B,A)$ forever, and $N<\infty$ almost surely.

Therefore,
$\SOTM_{\mathrm{SPRT}}(p_0,p_1,\varepsilon)$ halts almost surely under
every $p\in[0,1]$.
\end{proof}

\begin{theorem}[Reliability guarantee]\label{thm:sprt-reliability}
The certification SOTM $\SOTM_{\mathrm{SPRT}}(p_0,p_1,\varepsilon)$ is
$\varepsilon$-error-bounded on the certifiable regions $p\ge p_1$ and
$p\le p_0$.
\end{theorem}

\begin{proof}
By Theorem~\ref{thm:sprt-halting}, the certification SOTM halts almost surely
under every $p\in[0,1]$. Recall that the scores $Z_1,Z_2,\ldots$ are i.i.d.\
random variables with Bernoulli distribution $\Ber(p)$.

Let $\mathcal{Z}_{\mathrm{stop}}$ denote the set of possible stopped score
sequences. For a sequence
$\mathbf{z}=(z_1,\ldots,z_n)\in\mathcal{Z}_{\mathrm{stop}}$, write
\[
  \Pr_p(\mathbf{Z}^{\mathrm{stop}}=\mathbf{z})
  :=
  \Pr_p\bigl(N=n,Z_1=z_1,\ldots,Z_n=z_n\bigr)
\]
for the probability that the SOTM stops with score sequence $\mathbf{z}$.
Its stopped log-likelihood ratio is
\[
  L_n(\mathbf{z})
  =
  \sum_{i=1}^{n}
  \ln\frac{\Ber(p_1)(z_i)}{\Ber(p_0)(z_i)}.
\]
Because the SPRT stopping rule is determined by the score sequence, each
$\mathbf{z}\in\mathcal{Z}_{\mathrm{stop}}$ specifies the event that the
SOTM stops at turn $n$ with scores $\mathbf{z}$.

By independence,
\[
  \Pr_{p_1}(\mathbf{Z}^{\mathrm{stop}}=\mathbf{z})
  =
  \prod_{i=1}^{n}\Ber(p_1)(z_i),
\]
and similarly
\[
  \Pr_{p_0}(\mathbf{Z}^{\mathrm{stop}}=\mathbf{z})
  =
  \prod_{i=1}^{n}\Ber(p_0)(z_i).
\]
Therefore,
\[
  \frac{
    \Pr_{p_1}(\mathbf{Z}^{\mathrm{stop}}=\mathbf{z})
  }{
    \Pr_{p_0}(\mathbf{Z}^{\mathrm{stop}}=\mathbf{z})
  }
  =
  \prod_{i=1}^{n}
  \frac{\Ber(p_1)(z_i)}{\Ber(p_0)(z_i)}
  =
  e^{L_n(\mathbf{z})}.
\]
Equivalently,
\[
  \Pr_{p_1}(\mathbf{Z}^{\mathrm{stop}}=\mathbf{z})
  =
  e^{L_n(\mathbf{z})}
  \Pr_{p_0}(\mathbf{Z}^{\mathrm{stop}}=\mathbf{z}).
\]

Let $E$ be an event depending only on the stopped score sequence. Then
\begin{align}
  \Pr_{p_1}(E)
  &=
  \sum_{\mathbf{z}\in E}
  e^{L_n(\mathbf{z})}
  \Pr_{p_0}(\mathbf{Z}^{\mathrm{stop}}=\mathbf{z}) \notag\\
  &=
  \E_{p_0}\!\left[\mathbf{1}_E e^{L_N}\right].
  \label{eq:identity}
\end{align}

If $p\ge p_1$, the correct verdict is $\mathsf{Reliable}$, and the error
event is $\{L_N\le B\}$. We first prove the error bound at $p=p_1$.

Apply identity~\eqref{eq:identity} to $E=\{L_N\le B\}$. Since
$e^{L_N}\le e^B$ on this event,
\[
  \Pr_{p_1}(L_N\le B)
  =
  \E_{p_0}\!\left[\mathbf{1}_{\{L_N\le B\}}e^{L_N}\right]
  \le
  e^B\Pr_{p_0}(L_N\le B)
  \le
  e^B.
\]
Since $B=-A$ and
\[
  A=\ln\frac{1-\varepsilon/2}{\varepsilon/2},
\]
we have
\[
  e^B
  =
  \frac{\varepsilon/2}{1-\varepsilon/2}
  \le
  \varepsilon.
\]
Hence
\[
  \Pr_{p_1}(L_N\le B)\le\varepsilon.
\]

It remains to extend this bound to every $p\ge p_1$. To compare the score
distributions under $p$ and $p_1$ pointwise, we represent them on a common
probability space using the same underlying random draws. Let
$U_1,U_2,\ldots$ be independent uniform random variables on $[0,1]$. For
each $\pi\in[0,1]$, define
\[
  Z_i(\pi)=\mathbf{1}_{\{U_i\le\pi\}}.
\]
Then $Z_i(\pi)\sim\Ber(\pi)$. Since $p\ge p_1$, we have
\[
  \begin{array}{c|cc}
    \text{condition} & Z_i(p_1) & Z_i(p) \\
    \hline
    U_i\le p_1 & 1 & 1 \\
    p_1<U_i\le p & 0 & 1 \\
    U_i>p & 0 & 0
  \end{array}
\]
and hence
\[
  Z_i(p)\ge Z_i(p_1)
  \qquad\text{for every }i.
\]
Let $L_m(\pi)$ denote the log-likelihood ratio constructed from
$Z_1(\pi),\ldots,Z_m(\pi)$. Since the log-likelihood increment is larger
when $Z_i=1$ than when $Z_i=0$,
\[
  L_m(p)\ge L_m(p_1)
  \qquad\text{for every }m.
\]
For each $\pi\in[0,1]$, let
\[
  N_\pi
  :=
  \inf\{m\ge1:L_m(\pi)\notin(B,A)\}
\]
be the first exit time of the coupled log-likelihood ratio. Suppose that the
SOTM with parameter~$p$ halts with verdict $\mathsf{Unreliable}$,
equivalently, that
\[
  L_{N_p}(p)\le B.
\]
For every $m<N_p$, we have $L_m(p)<A$. Since
$L_m(p_1)\le L_m(p)$, the SOTM with parameter~$p_1$ cannot halt with verdict
$\mathsf{Reliable}$ before turn $N_p$. At turn $N_p$,
\[
  L_{N_p}(p_1)
  \le
  L_{N_p}(p)
  \le
  B.
\]
Thus the SOTM with parameter~$p_1$ halts with verdict
$\mathsf{Unreliable}$ no later than turn $N_p$. 
Let
\[
  E_\pi^-
  :=
  \{\text{the SOTM with parameter $\pi$ halts with verdict
  $\mathsf{Unreliable}$}\}.
\]
The preceding argument shows that, for every realization of the shared
uniform random variables,
\[
  E_p^-\subseteq E_{p_1}^-.
\]
Indeed, whenever the SOTM with parameter~$p$ halts with verdict
$\mathsf{Unreliable}$, the SOTM with parameter~$p_1$ has reached the lower
boundary by that same turn, possibly earlier, and therefore also halts with
verdict $\mathsf{Unreliable}$. Taking probabilities yields
\[
  \Pr_p(v=\mathsf{Unreliable})
  \le
  \Pr_{p_1}(v=\mathsf{Unreliable})
  \le
  \varepsilon
  \qquad\text{for every }p\ge p_1.
\]

The argument for $p\le p_0$ is symmetric. If $p\le p_0$, the correct
verdict is $\mathsf{Unreliable}$, and the error event is $\{L_N\ge A\}$.
Interchanging the roles of $p_0$ and $p_1$ gives
\[
  \Pr_{p_0}(L_N\ge A)
  \le
  e^{-A}
  =
  \frac{\varepsilon/2}{1-\varepsilon/2}
  \le
  \varepsilon.
\]
The same monotonicity coupling yields
\[
  \Pr_p(v=\mathsf{Reliable})
  =
  \Pr_p(L_N\ge A)
  \le
  \Pr_{p_0}(L_N\ge A)
  \le
  \varepsilon
  \qquad\text{for every }p\le p_0.
\]
Thus the certification SOTM is $\varepsilon$-error-bounded on the
certifiable regions $p\ge p_1$ and $p\le p_0$.
\end{proof}

To bound the expected sample size, assume that $N<\infty$. Define the upper- and lower-boundary overshoots
by
\[
  O_A=(L_N-A)\mathbf{1}_{\{L_N\ge A\}},
  \qquad
  O_B=(B-L_N)\mathbf{1}_{\{L_N\le B\}}.
\]
Since the SOTM halts through at most one boundary,
\[
  O:=\max\{O_A,O_B\}=O_A+O_B.
\]
The largest upward increment of $L_i$ is
\[
  \ln\frac{p_1}{p_0},
\]
and the magnitude of the largest downward increment is
\[
  \ln\frac{1-p_0}{1-p_1}.
\]
Therefore,
\[
  0\le O\le\Delta(p_0,p_1),
\]
where
\[
  \Delta(p_0,p_1)
  :=
  \max\!\left\{
    \ln\frac{p_1}{p_0},
    \ln\frac{1-p_0}{1-p_1}
  \right\}.
\]

Define the normalized overshoot bounds
\[
  C_1(p_0,p_1)
  :=
  \frac{\Delta(p_0,p_1)}{\KL{p_1}{p_0}},
  \qquad
  C_0(p_0,p_1)
  :=
  \frac{\Delta(p_0,p_1)}{\KL{p_0}{p_1}}.
\]
They are the worst-case overshoot bounds normalized by the relevant drifts.

For the certification SOTM
$\SOTM_{\mathrm{SPRT}}(p_0,p_1,\varepsilon)$, define the per-turn token cost
\[
  W_i^{\alpha,\beta}
  =
  \alpha|\tau(x_i)|+\beta|\tau(r_i)|.
\]
Since the query--response turns are i.i.d., the expected per-turn token cost
is the same at every turn. For $p\in[0,1]$, write
\[
  \bar c_{p,\alpha,\beta}
  :=
  \E_p[W_1^{\alpha,\beta}].
\]

\begin{theorem}[Expected sample size]\label{thm:sprt-sample}
For the certification SOTM
$\SOTM_{\mathrm{SPRT}}(p_0,p_1,\varepsilon)$,
\begin{align}
  \E_p[N]
  &\le
  \frac{
    \ln\!\bigl(\tfrac{1-\varepsilon/2}{\varepsilon/2}\bigr)
  }{\KL{p_1}{p_0}}
  +
  C_1(p_0,p_1),
  \qquad \text{for every } p\ge p_1.
  \label{eq:sample} \\
  \E_p[N]
  &\le
  \frac{
    \ln\!\bigl(\tfrac{1-\varepsilon/2}{\varepsilon/2}\bigr)
  }{\KL{p_0}{p_1}}
  +
  C_0(p_0,p_1),
  \qquad \text{for every } p\le p_0.
  \label{eq:sample-symmetric}
\end{align}
\end{theorem}

\begin{proof}
We first prove the expected sample-size bound. Fix $p\ge p_1$.
For $t\ge1$, let $N_t:=N\wedge t=\min\{N,t\}$.
Recall that the log-likelihood increments
\[
  \ell_i
  =
  Z_i\ln\frac{p_1}{p_0}
  +(1-Z_i)\ln\frac{1-p_1}{1-p_0}
\]
are i.i.d.\ under~$\Pr_p$, with mean
\[
  f(p)
  =
  \E_p[\ell_i]
  =
  p\ln\frac{p_1}{p_0}
  +(1-p)\ln\frac{1-p_1}{1-p_0}.
\]
Since $f(p)$ is increasing in~$p$ (see the proof of Theorem \ref{thm:sprt-reliability}),
\[
  f(p)
  \ge
  f(p_1)
  =
  \KL{p_1}{p_0}
  \qquad\text{for every }p\ge p_1.
\]

Since $N_t$ is a bounded stopping time, Wald's identity gives
\[
  \E_p[L_{N_t}]
  =
  \E_p[N_t]f(p).
\]
If $N\le t$, then
\[
  L_{N_t}
  =
  L_N
  \le
  A+O_A
  \le
  A+\Delta(p_0,p_1).
\]
If $N>t$, then $L_{N_t}=L_t\in(B,A)$. Thus, in all cases,
\[
  L_{N_t}\le A+\Delta(p_0,p_1).
\]
Consequently,
\[
  \E_p[N_t]
  =
  \frac{\E_p[L_{N_t}]}{f(p)}
  \le
  \frac{A+\Delta(p_0,p_1)}{\KL{p_1}{p_0}}
  =
  \frac{A}{\KL{p_1}{p_0}}
  +
  C_1(p_0,p_1).
\]
Letting $t\to\infty$ and applying monotone convergence yields
\[
  \E_p[N]
  \le
  \frac{A}{\KL{p_1}{p_0}}
  +
  C_1(p_0,p_1).
\]
Since
\[
  A
  =
  \ln\frac{1-\varepsilon/2}{\varepsilon/2},
\]
this proves the sample-size bound~\eqref{eq:sample}.

Now fix $p\le p_0$. Since $f(p)$ is increasing,
\[
  -f(p)
  \ge
  -f(p_0)
  =
  \KL{p_0}{p_1}.
\]
If $N\le t$, then
\begin{align*}
L_{N_t} &= L_N = B - O_B, \\
  -L_{N_t}
  &=
  -B+O_B
  \le
  A+\Delta(p_0,p_1).
\end{align*}
If $N>t$, then $L_{N_t}=L_t\in(B,A)$, and hence
\[
  -L_{N_t}< -B=A.
\]
Therefore,
\[
  -L_{N_t}\le A+\Delta(p_0,p_1).
\]
Applying Wald's identity to the bounded stopping time $N_t$ gives
\[
  \E_p[-L_{N_t}]
  =
  -f(p)\,\E_p[N_t].
\]
Since
\[
  -f(p)
  \ge
  -f(p_0)
  =
  \KL{p_0}{p_1},
\]
we obtain
\[
  \E_p[N_t]
  =
  \frac{\E_p[-L_{N_t}]}{-f(p)}
  \le
  \frac{A+\Delta(p_0,p_1)}{\KL{p_0}{p_1}}
  =
  \frac{A}{\KL{p_0}{p_1}}
  +
  C_0(p_0,p_1).
\]
Letting $t\to\infty$ and applying monotone convergence proves the sample-size bound
\eqref{eq:sample-symmetric}.
\end{proof}

\begin{theorem}[Expected token cost]\label{thm:sprt-cost}
For the certification SOTM
$\SOTM_{\mathrm{SPRT}}(p_0,p_1,\varepsilon)$,
\begin{align}
  \E_p[\TOK]
  &=
  \bar c_{p,\alpha,\beta}\,\E_p[N]. \label{eq:equality}\\
  \E_p[\TOK]
  &\le
  \bar c_{p,\alpha,\beta}
  \left(
    \frac{
      \ln\!\bigl(\tfrac{1-\varepsilon/2}{\varepsilon/2}\bigr)
    }{\KL{p_1}{p_0}}
    +
    C_1(p_0,p_1)
  \right),
  \qquad \text{for every } p\ge p_1.
  \label{eq:cost1} \\
  \E_p[\TOK]
  &\le
  \bar c_{p,\alpha,\beta}
  \left(
    \frac{
      \ln\!\bigl(\tfrac{1-\varepsilon/2}{\varepsilon/2}\bigr)
    }{\KL{p_0}{p_1}}
    +
    C_0(p_0,p_1)
  \right),
  \qquad \text{for every } p\le p_0.
  \label{eq:cost0}
\end{align}
\end{theorem}

\begin{proof}
Since the per-turn token costs are i.i.d.\ and bounded, and $N$ is a
stopping time with $\E_p[N]<\infty$ by Theorem~\ref{thm:sprt-sample}, Wald's identity gives the expected token cost~(\ref{eq:equality}):
\[
  \E_p[\TOK]
  =
  \E_p\!\left[\sum_{n=1}^{N}W_n^{\alpha,\beta}\right]
  =
  \bar c_{p,\alpha,\beta}\,\E_p[N].
\]
Combining this identity with the sample-size bound~\eqref{eq:sample} proves the cost bound~\eqref{eq:cost1},
and combining it with the sample-size bound~\eqref{eq:sample-symmetric} proves the cost bound~\eqref{eq:cost0}.
\end{proof}

\begin{corollary}[Small-gap bounds]\label{cor:small-gap}
For the certification SOTM
$\SOTM_{\mathrm{SPRT}}(p_0,p_1,\varepsilon)$, let $\Delta=p_1-p_0$.
For small $\Delta$, the Bernoulli divergences satisfy
\begin{align*}
  \KL{p_1}{p_0}
  &=
  \frac{\Delta^2}{2p_0(1-p_0)}
  +
  O(\Delta^3), \\
\KL{p_0}{p_1}
  &=
  \frac{\Delta^2}{2p_1(1-p_1)}
  +
  O(\Delta^3).
\end{align*}
Consequently, for $p\ge p_1$, the expected sample-size bound has the approximation
\[
  \E_p[N]
  \lesssim
  \frac{
    2p_0(1-p_0)
    \ln\!\bigl(\tfrac{1-\varepsilon/2}{\varepsilon/2}\bigr)
  }{\Delta^2}
  +
  C(p_0,p_1),
\]
and the corresponding expected token-cost bound has the approximation
\[
  \E_p[\TOK]
  \lesssim
  \bar c_p
  \left(
    \frac{
      2p_0(1-p_0)
      \ln\!\bigl(\tfrac{1-\varepsilon/2}{\varepsilon/2}\bigr)
    }{\Delta^2}
    +
    C(p_0,p_1)
  \right).
\]

Similarly, for $p\le p_0$,
\[
  \E_p[N]
  \lesssim
  \frac{
    2p_1(1-p_1)
    \ln\!\bigl(\tfrac{1-\varepsilon/2}{\varepsilon/2}\bigr)
  }{\Delta^2}
  +
  C(p_1,p_0),
\]
and
\[
  \E_p[\TOK]
  \lesssim
  \bar c_p
  \left(
    \frac{
      2p_1(1-p_1)
      \ln\!\bigl(\tfrac{1-\varepsilon/2}{\varepsilon/2}\bigr)
    }{\Delta^2}
    +
    C(p_1,p_0)
  \right).
\]
Thus, in both certifiable regions, the dominant sample and token costs
scale as $\Delta^{-2}$.
\end{corollary}

\begin{proof}[Proof Sketch]
The approximation comes from a second-order Taylor expansion of the Bernoulli KL divergence in $\Delta$ around $\Delta = 0$.
\end{proof}

\begin{remark}[Application]
To determine whether an LLM meets a target reliability level
$p_{\min}$ for a given application domain, set $p_1=p_{\min}$ and
choose $p_0<p_{\min}$ as an acceptable rejection margin. If the
certification SOTM outputs $\mathsf{Reliable}$, then the oracle is
certified as meeting the target reliability level, with error probability
at most~$\varepsilon$ on the certifiable region $p\ge p_1$. If it outputs
$\mathsf{Unreliable}$, then the oracle is certified as falling below the
baseline level, with error probability at most~$\varepsilon$ on the
certifiable region $p\le p_0$. The gap $(p_0,p_1)$ is a precision
parameter: narrowing it sharpens the distinction at the cost of more
queries.
\end{remark}

\begin{remark}[Computability]
The sample-size bounds in Theorem~\ref{thm:sprt-sample} can be estimated from
$\varepsilon,p_0,p_1$ before running the certification SOTM
$\SOTM_{\mathrm{SPRT}}(p_0,p_1,\varepsilon)$. The corresponding
expected token-cost bounds also depend on the one-turn expected token
cost $\bar c_{p,\alpha,\beta}$.
In the finite evaluation-domain and length-controlled response setting,
this quantity is finite and can be estimated from the evaluation set, the
chosen token-cost parameters, and a small pilot sample of oracle
responses before certification begins.
\end{remark}

\begin{remark}[Sample size cap]
Although the certification SOTM
$\SOTM_{\mathrm{SPRT}}(p_0,p_1,\varepsilon)$ halts almost surely, a finite
implementation may still impose a maximum turn budget~$N_{\max}$. If the
certification SOTM has not halted by turn~$N_{\max}$, the implementation may return a
third verdict $\mathsf{Inconclusive}$. This verdict is outside the
two-sided error guarantee: the guarantee in
Theorem~\ref{thm:sprt-reliability} applies to runs that return either
$\mathsf{Reliable}$ or $\mathsf{Unreliable}$. The expected number of
turns in Theorem~\ref{thm:sprt-sample} provides a principled basis for
choosing~$N_{\max}$; in practice, setting $N_{\max}$ to a suitable
multiple of the expected sample-size bound makes inconclusive outcomes
rare when the oracle lies away from the ambiguity gap.
\end{remark}

\begin{remark}[Budget caveats]
Two caveats apply when using Theorem~\ref{thm:sprt-cost} for
provisioning. First, the theorem gives an \emph{expected} 
token-cost bound, not a deterministic worst-case budget. A fixed
implementation budget should therefore include additional slack if one
wants inconclusive outcomes to be rare. Second, the stated bounds are
designed for the certifiable regions $p\ge p_1$ and $p\le p_0$. When the
true reliability lies inside the ambiguity gap $(p_0,p_1)$, the SOTM is
not required to certify the oracle in either direction, and the stopping
time can be substantially larger. This reflects the intrinsic difficulty of distinguishing two nearby
Bernoulli reliability levels: when $p_1-p_0$ is small, each observation
carries little information about which side of the certification problem
the oracle lies on.
\end{remark}

\begin{remark}[Attribution]
The sample complexity of distinguishing two Bernoulli reliability levels
$p_0$ and $p_1$ is classical, going back to Wald's analysis of the sequential probability
ratio test and to information-theoretic testing bounds such as the
Chernoff--Stein lemma and Le~Cam's two-point method
\cite{Wald1947,CoverThomas2006,Tsybakov2009}. 
\end{remark}


\section{A Matching Reliability Certification Lower Bound}\label{sec:lower}

Theorem~\ref{thm:sprt-cost} gives an upper bound on the sample size and token
costs of reliability certification. We now prove a lower bound in a general setting that permits the
certification SOTM to choose each query adaptively from its previous
interaction history. The lower bound assumes that each query provides at most a uniformly bounded
amount of information for distinguishing the two boundary hypotheses. This assumption prevents any single query from providing arbitrarily
strong evidence for making that distinction.

\begin{definition}[Adaptive certification SOTM]
\label{def:adaptive-certification-model}
Let $H_0$ and $H_1$ be two hypotheses such that,
for each $b\in\{0,1\}$ and query $x\in X$, the oracle response has
distribution $P_{b,x}$ on $\alphabet^*$. Assume that
\[
  \E_{x\sim\mathcal D_X,\ r\sim P_{b,x}}[S(x,r)]
  =
  p_b
  \qquad\text{for }b\in\{0,1\}.
\]

An \emph{adaptive certification SOTM} proceeds sequentially. At turn $n$,
it selects a query $x_n\in X$ using its internal randomness $U$ and the
previous transcript $(x_1,r_1,\ldots,x_{n-1},r_{n-1})$.
Conditional on $x_n$, the oracle returns a response
\[
  r_n\sim P_{b,x_n}
\]
independently of the previous turns under hypothesis $H_b$.
The SOTM then
computes the score
\[
  Z_n=S(x_n,r_n).
\]
Its stopping rule and final verdict may depend on the transcript and its
internal randomness.
\end{definition}

\begin{theorem}[Adaptive certification lower bound]
\label{thm:adaptive-certification-lower}
Fix two hypotheses $H_0$ and $H_1$, where $H_b$ indexes an oracle $\oracle$ with
reliability $p_b$ for $b\in\{0,1\}$. Let
$\SOTM=(M,\oracle)$ be an adaptive certification SOTM that is
$\varepsilon$-error-bounded in the sense of
Definition~\ref{def:error-bounded}, for some $\varepsilon\in(0,1/2)$.
Assume that there is a common constant $I_{\max}<\infty$ such that
\[
  \KL{P_{1,x}}{P_{0,x}}
  \le
  I_{\max}
  \qquad\text{for every }x\in X.
\]
Let $N$ be the number of turns before the SOTM halts. Then
\[
  \E_{H_1}[N]
  \ge
  \frac{\KL{1-\varepsilon}{\varepsilon}}{I_{\max}}.
\]
Consequently, as $\varepsilon\to0$,
\[
  \E_{H_1}[N]
  \ge
  (1+o(1))
  \frac{\ln(1/\varepsilon)}{I_{\max}}.
\]
\end{theorem}

\begin{proof}
If $\E_{H_1}[N]=\infty$, the claim is immediate. We therefore assume that
$\E_{H_1}[N]<\infty$.

Let $U$ collect the certification SOTM's internal randomness, independently
of the oracle. For each turn $i$, let $X_i$ and $R_i$ denote the random
query selected by the SOTM and the resulting random oracle response,
respectively. Consider the augmented stopped transcript
\[
  \widetilde{\mathcal T}_N
  =
  (U,X_1,R_1,\ldots,X_N,R_N).
\]
For any random object $\Xi$, let $\mathcal L_{H_b}(\Xi)$ denote the
probability distribution of $\Xi$ under hypothesis $H_b$. Write
\[
  P
  =
  \mathcal L_{H_1}(\widetilde{\mathcal T}_N),
  \qquad
  Q
  =
  \mathcal L_{H_0}(\widetilde{\mathcal T}_N).
\]

For a realized augmented stopped transcript
\[
  \boldsymbol t=(u,x_1,r_1,\ldots,x_N,r_N),
\]
its probability under $H_b$ factors as
\[
  \Pr_{H_b}(\widetilde{\mathcal T}_N=\boldsymbol t)
  =
  \Pr(U=u)
  \prod_{n=1}^{N}
  \Pr(x_n\mid u,x_1,r_1,\ldots,x_{n-1},r_{n-1})
  P_{b,x_n}(r_n).
\]
The internal-randomness factor and the query-selection factors are the same
under both hypotheses, because the certification SOTM follows the same
algorithm under $H_0$ and $H_1$. Hence
\[
  \frac{P(\boldsymbol t)}{Q(\boldsymbol t)}
  =
  \prod_{n=1}^{N}
  \frac{P_{1,x_n}(r_n)}{P_{0,x_n}(r_n)}.
\]
Taking logarithms gives the stopped log-likelihood ratio
\[
  L_N
  =
  \ln\frac{P(\boldsymbol t)}{Q(\boldsymbol t)}
  =
  \sum_{n=1}^{N}
  \ln\frac{P_{1,x_n}(r_n)}{P_{0,x_n}(r_n)}.
\]
Because $L_N$ is the log-likelihood ratio of the augmented stopped
transcript under $H_1$ relative to $H_0$,
\[
  \E_{H_1}[L_N]
  =
  \KL{P}{Q}.
\]

Conditional on the previous augmented transcript and the chosen query
$X_n$, the response $R_n$ has distribution $P_{1,X_n}$ under $H_1$.
Therefore,
\[
  \E_{H_1}\!\left[
    \ln\frac{P_{1,X_n}(R_n)}{P_{0,X_n}(R_n)}
    \,\middle|\,
    \widetilde{\mathcal T}_{n-1},X_n
  \right]
  =
  \KL{P_{1,X_n}}{P_{0,X_n}}.
\]
Applying the sequential chain rule for KL divergence to the transcript
stopped at $N\wedge t$, and then letting $t\to\infty$, gives
\[
  \E_{H_1}[L_N]
  =
  \E_{H_1}\!\left[
    \sum_{n=1}^{N}
    \KL{P_{1,X_n}}{P_{0,X_n}}
  \right].
\]
The per-query information bound therefore implies
\[
  \E_{H_1}[L_N]
  \le
  I_{\max}\E_{H_1}[N].
\]

We use the data-processing inequality for KL divergence, which states that
applying the same post-processing function to data cannot increase the
distinguishability of two probability distributions. More precisely, if
$P$ and $Q$ are probability distributions and $g$ is a measurable function,
then
\[
  \KL{P}{Q}
  \ge
  \KL{P\circ g^{-1}}{Q\circ g^{-1}};
\]
see, for example, \cite[Section~2.8]{CoverThomas2006}.

Let
\[
  g:\{\text{augmented stopped transcripts}\}
  \to
  \{\mathsf{Reliable},\mathsf{Unreliable}\}
\]
denote the certification SOTM's decision rule, so that
\[
  g(\widetilde{\mathcal T}_N)=v.
\]
For example,
\[
  g^{-1}(\{\mathsf{Reliable}\})
  =
  \{\boldsymbol{t}:g(\boldsymbol{t})=\mathsf{Reliable}\}
\]
is the set of augmented stopped transcripts that produce the verdict
$\mathsf{Reliable}$. Thus,
\[
  P\circ g^{-1}
  =
  \mathcal L_{H_1}(v),
  \qquad
  Q\circ g^{-1}
  =
  \mathcal L_{H_0}(v).
\]
Applying the data-processing inequality with the post-processing function
$g$ gives
\[
  \KL{P}{Q}
  \ge
  \KL{
    \mathcal L_{H_1}(v)
  }{
    \mathcal L_{H_0}(v)
  }.
\]
Because the verdict is binary, its distribution under each hypothesis is a
Bernoulli distribution after identifying $\mathsf{Reliable}$ with $1$ and
$\mathsf{Unreliable}$ with $0$.

Let
\[
  q_1=\Pr_{H_1}(v=\mathsf{Reliable}),
  \qquad
  q_0=\Pr_{H_0}(v=\mathsf{Reliable}).
\]
Thus,
\[
  \KL{
    \mathcal L_{H_1}(v)
  }{
    \mathcal L_{H_0}(v)
  }
  =
  \KL{q_1}{q_0}.
\]
The error guarantee gives
\[
  q_1\ge1-\varepsilon,
  \qquad
  q_0\le\varepsilon.
\]
Since Bernoulli KL divergence is increasing in its first argument and
decreasing in its second argument when the first argument exceeds the
second,
\[
  \KL{q_1}{q_0}
  \ge
  \KL{1-\varepsilon}{\varepsilon}.
\]
Combining the preceding inequalities yields
\[
  I_{\max}\E_{H_1}[N]
  \ge
  \E_{H_1}[L_N]
  \ge
  \KL{1-\varepsilon}{\varepsilon}.
\]
This proves the non-asymptotic lower bound.

Finally,
\[
  \KL{1-\varepsilon}{\varepsilon}
  =
  (1-2\varepsilon)
  \ln\frac{1-\varepsilon}{\varepsilon}
  =
  (1+o(1))\ln\frac{1}{\varepsilon}
\]
as $\varepsilon\to0$. Substituting this expression into the
non-asymptotic lower bound proves the asymptotic statement.
\end{proof}

When
\[
  I_{\max}=\KL{p_1}{p_0},
\]
the lower bound matches the leading-order $\ln(1/\varepsilon)$ dependence
of the sample-size upper bound for the SPRT-based certification SOTM in Theorem~\ref{thm:sprt-cost}.

\begin{corollary}[Bernoulli certification lower bound]
\label{cor:bernoulli-certification-lower}
Suppose that under reliability
hypothesis $H_b$, a query produces a correctness score distributed as
$Z\sim\Ber(p_b)$, $b\in\{0,1\}$.
Then every $\varepsilon$-error-bounded certification SOTM satisfies
\[
  \E_{p_1}[N]
  \ge
  \frac{\KL{1-\varepsilon}{\varepsilon}}{\KL{p_1}{p_0}}.
\]
Consequently, as $\varepsilon\to0$,
\[
  \E_{p_1}[N]
  \ge
  (1+o(1))
  \frac{\ln(1/\varepsilon)}{\KL{p_1}{p_0}}.
\]
If additionally $\Delta=p_1-p_0$ is small and $p_0,p_1$ remain bounded
away from $0$ and $1$, then
\[
  \KL{p_1}{p_0}
  =
  \frac{\Delta^2}{2p_0(1-p_0)}
  +
  O(\Delta^3).
\]
Therefore, as $\varepsilon\to0$ and $\Delta\to0$,
\[
  \E_{p_1}[N]
  \ge
  (1+o(1))
  \frac{
    2p_0(1-p_0)\ln(1/\varepsilon)
  }{\Delta^2}.
\]
\end{corollary}

\begin{proof}
It suffices to show that $I_{\max}=\KL{p_1}{p_0}$.
Suppose that, for every query $x\in X$, the observed per-query variable is
the binary correctness score and, under hypothesis $H_b$, it has
distribution
\[
  P_{b,x}
  =
  \Ber(p_b),
  \qquad
  b\in\{0,1\}.
\]
Then
\[
  \KL{P_{1,x}}{P_{0,x}}
  =
  \KL{\Ber(p_1)}{\Ber(p_0)}
  =
  \KL{p_1}{p_0}
  \qquad\text{for every }x\in X.
\]
Thus the common per-query information bound may be taken as
\[
  I_{\max}
  =
  \KL{p_1}{p_0}.
\]
\end{proof}

\begin{corollary}[Tightness for small error probability]
The upper and lower bounds match in their leading dependence on the
error probability. 
\end{corollary}
\begin{proof}
The SPRT-based certification SOTM in
Theorem~\ref{thm:sprt-cost} satisfies, for $p\ge p_1$,
\[
  \E_p[N]
  \le
  \frac{
    \ln\!\bigl(\tfrac{1-\varepsilon/2}{\varepsilon/2}\bigr)
  }{\KL{p_1}{p_0}}
  +
  C(p_0,p_1).
\]
On the other hand, the Bernoulli lower bound gives
\[
  \E_{p_1}[N]
  \ge
  (1+o(1))
  \frac{\ln(1/\varepsilon)}{\KL{p_1}{p_0}}
  \qquad
  \text{as } \varepsilon\to0 .
\]
Since
\[
  \ln\!\bigl(\tfrac{1-\varepsilon/2}{\varepsilon/2}\bigr)
  =
  \ln(1/\varepsilon)+O(1),
\]
the upper bound for the SPRT-based certification SOTM and the lower bound for any adaptive certification SOTM have the same leading term as
$\varepsilon\to0$. Thus, for fixed $p_0$ and $p_1$, the SPRT-based
certification SOTM is asymptotically optimal in its dependence on the
error probability. The same conclusion holds symmetrically on the
$p\le p_0$ side with $\KL{p_0}{p_1}$ in place of $\KL{p_1}{p_0}$.
\end{proof}

\begin{corollary}[Token complexity of reliability certification]
The token complexity of certifying the reliability of a given oracle $\oracle$ with respect to a given domain $\boldsymbol{d}$ is
\[
  \kappa_{\boldsymbol{d}}^{\mathrm{cert}}
  (p_0,p_1,\varepsilon;\alpha,\beta)
  =
  (1+o(1))\,
  \bar c(\alpha,\beta)\,
  \frac{\ln(1/\varepsilon)}{\KL{p_1}{p_0}}
  \qquad
  \text{as } \varepsilon\to0,
\]
for certifying the target $p \geq p_1$ against the below-threshold $p\leq p_0$,
where $\bar c(\alpha,\beta)$ is the expected per-turn token cost. The symmetric
statement for certifying the below-threshold boundary case against the
target boundary case uses $\KL{p_0}{p_1}$ in place of
$\KL{p_1}{p_0}$.

When $\Delta=p_1-p_0$ is small, 
\[
  \kappa_{\boldsymbol{d}}^{\mathrm{cert}}
  (p_0,p_1,\varepsilon;\alpha,\beta)
  =
  (1+o(1))\,
  \bar c(\alpha,\beta)\,
  \frac{\ln(1/\varepsilon)}{\Delta^2},
\]
up to the constant factor determined by the reliability level. In this
sense, reliability certification becomes only logarithmically more
expensive as the desired error probability decreases, but quadratically
more expensive as the reliability gap narrows. The SPRT-based
certification SOTM is therefore asymptotically token-optimal for
Bernoulli reliability certification in the small-error regime.
\end{corollary}

\section{Conclusion}\label{sec:conclusion}

This paper studies reliability certification for stochastic oracles in
the SOTM framework. Given an evaluation domain
$\boldsymbol{d}=(X,S,\mathcal{D}_X)$ and reliability thresholds
$0<p_0<p_1<1$, we showed that certification can be formulated as a
sequential decision problem: determine, with controlled error probability,
whether an oracle meets the target reliability level $p_1$ or falls below
the lower threshold $p_0$, while making no required decision in the
ambiguity gap $(p_0,p_1)$.

Our main upper-bound construction is an SPRT-based certification SOTM.
It repeatedly queries the oracle, computes binary correctness scores, and
stops once the accumulated log-likelihood evidence crosses one of two
thresholds. We proved that this procedure halts almost surely, satisfies
the desired two-sided reliability guarantee on the certifiable regions,
and gives an explicit upper bound on the certification token complexity.
This upper bound is computable from the reliability thresholds, the error
probability, and the expected per-turn token cost. In particular, for
fixed thresholds, the dominant dependence on the error probability is
logarithmic, while for a small threshold gap $\Delta=p_1-p_0$, the
dominant dependence on the gap is quadratic in $\Delta^{-1}$.

We also proved a matching lower bound for adaptive certification
SOTMs. Even when a certification SOTM may choose queries
adaptively, each query can contribute only a bounded amount of
statistical information about which certification hypothesis is
underlying the oracle's responses. Consequently, any procedure that
distinguishes the below-threshold boundary case from the target boundary
case with error at most $\varepsilon$ must have certification token
complexity at least on the order of $\ln(1/\varepsilon)$ divided by the
per-query information, up to the relevant per-turn token scale. In the
Bernoulli certification specialization, this lower bound matches the
leading term of the SPRT upper bound as $\varepsilon\to0$.

Together, the upper and lower bounds characterize the token complexity of
oracle reliability certification in the small-error regime. For
certifying the target boundary case $p=p_1$ against the below-threshold
boundary case $p=p_0$, the leading dependence is governed by
\[
  \kappa_{\boldsymbol{d}}^{\mathrm{cert}}
  (p_0,p_1,\varepsilon;\alpha,\beta)
  \approx
  \bar c(\alpha,\beta)\,
  \frac{\ln(1/\varepsilon)}{\KL{p_1}{p_0}},
\]
where $\bar c(\alpha,\beta)$ denotes the expected per-turn token cost. Thus,
certification becomes only logarithmically more demanding as the desired
error probability decreases, but substantially more demanding as the
reliability gap narrows. When $\Delta=p_1-p_0$ is small, the leading
dependence is proportional to
\[
  \bar c(\alpha,\beta)\,\frac{\ln(1/\varepsilon)}{\Delta^2}.
\]

These results give reliability certification a token-complexity
interpretation. In particular, certification token complexity is the
minimum expected token cost required to evaluate an existing oracle on a
specified domain with controlled error probability. This cost is governed
by the statistical information obtained per query and by the desired
separation between acceptable and unacceptable reliability levels. The
resulting framework provides a principled way to estimate the token
complexity of certifying an LLM or other stochastic oracle before
deployment in a specialized application domain.

Several directions remain open. One natural extension is to move beyond
binary scoring rules and study graded or task-dependent scoring
functions. Although a scoring function
with values in $[0,1]$ can be reduced to the present setting by randomized
Bernoulliization, this reduction may discard information contained in the
score's full distribution. It is more interesting to develop certification SOTMs that directly use continuous scores while retaining
finite-sample error guarantees and sharp token-complexity bounds. In
particular, it would be valuable to characterize the optimal leading-order
token complexity for testing whether the mean score lies above or below the
reliability thresholds when the score distribution is not assumed to be
Bernoulli.

Another is to incorporate heterogeneous query costs, where
different regions of the evaluation domain have different token budgets
and different information values. A third direction is to analyze
certification under distribution shift, where the deployment distribution
may differ from the evaluation distribution. Finally, the adaptive lower
bound suggests that optimal certification should query not only for
representativeness, but also for information efficiency; understanding
that tradeoff is an important step toward practical, token-efficient
oracle evaluation.

\bibliographystyle{plain}
\bibliography{reference}

@article{Wang2026,
author       = {Wang, Jie},
title        = {Token Complexity Theory for {AI}-Augmented Computing},
journal      = {arXiv preprint arXiv:2606.12647},
year         = {2026},
url          = {https://arxiv.org/abs/2606.12647}
}

@inproceedings{Kudo2018,
author    = {Kudo, Taku and Richardson, John},
title     = {{SentencePiece}: A Simple and Language Independent Subword Tokenizer and Detokenizer for Neural Text Processing},
booktitle = {Proceedings of the 2018 Conference on Empirical Methods in Natural Language Processing: System Demonstrations},
pages     = {66--71},
publisher = {Association for Computational Linguistics},
year      = {2018}
}

@inproceedings{Schuster2012,
author    = {Schuster, Mike and Nakajima, Kaisuke},
title     = {Japanese and {Korean} Voice Search},
booktitle = {IEEE International Conference on Acoustics, Speech and Signal Processing},
pages     = {5149--5152},
publisher = {IEEE},
year      = {2012}
}

@inproceedings{Sennrich2016,
author    = {Sennrich, Rico and Haddow, Barry and Birch, Alexandra},
title     = {Neural Machine Translation of Rare Words with Subword Units},
booktitle = {Proceedings of the 54th Annual Meeting of the Association for Computational Linguistics},
pages     = {1715--1725},
publisher = {Association for Computational Linguistics},
year      = {2016}
}

@book{LehmannRomano2005,
author    = {Lehmann, E. L. and Romano, Joseph P.},
title     = {Testing Statistical Hypotheses},
publisher = {Springer},
edition   = {3},
year      = {2005}
}

@article{WaldWolfowitz1948,
author  = {Wald, Abraham and Wolfowitz, Jacob},
title   = {Optimum Character of the Sequential Probability Ratio Test},
journal = {Annals of Mathematical Statistics},
volume  = {19},
number  = {3},
pages   = {326--339},
year    = {1948}
}

@book{Wald1947,
author    = {Wald, Abraham},
title     = {Sequential Analysis},
publisher = {Wiley},
year      = {1947}
}

@book{CoverThomas2006,
author    = {Cover, Thomas M. and Thomas, Joy A.},
title     = {Elements of Information Theory},
publisher = {Wiley},
edition   = {2},
year      = {2006}
}

@book{Tsybakov2009,
author    = {Tsybakov, Alexandre B.},
title     = {Introduction to Nonparametric Estimation},
publisher = {Springer},
year      = {2009}
}

@book{Feller1971,
author    = {Feller, William},
title     = {An Introduction to Probability Theory and Its Applications},
volume    = {2},
edition   = {2},
publisher = {John Wiley \& Sons},
address   = {New York},
year      = {1971}
}

\end{document}